\documentclass[11pt]{article}

\usepackage{geometry}
\usepackage{amsthm}
\usepackage{amsmath}
\usepackage{amssymb}
\usepackage{hyperref}

\newtheorem{theorem}{Theorem}[section]

\newtheorem{definition}[theorem]{Definition}
\newtheorem{lemma}[theorem]{Lemma}

\newtheorem{claim}[theorem]{Claim}
\newtheorem{fact}[theorem]{Fact}

\newcommand\E{\mathop{\mathbf{E}}}
\newcommand{\eps}{\varepsilon}
\newcommand{\etal}{{\em et al.\ }}

\newcommand{\val}{\mathbf{val}}
\newcommand{\rpred}{random-predicate\ }
\newcommand{\x}{\mathbf{x}}
\newcommand{\y}{\mathbf{y}}
\newcommand{\z}{\mathbf{z}}
\renewcommand{\v}{\mathbf{v}}

\title{Parallel Repetition of $k$-Player Projection Games}
\author{}
\author{Amey Bhangale\thanks{Department of Computer Science and Engineering, University of California, Riverside. Supported by the Hellman Fellowship award.}
	\and
    Mark Braverman\thanks{Department of Computer Science at Princeton University. Supported by the National Science Foundation under the Alan T. Waterman Award, Grant No. 1933331}
    \and
	Subhash Khot\thanks{Department of Computer Science, Courant Institute of Mathematical Sciences, New York University. Supported by the NSF Award CCF-2130816 and the Simons Investigator Award.}
	\and
    Yang P. Liu\thanks{School of Mathematics, Institute for Advanced Study, Princeton, NJ. Partially supported by NSF DMS-1926686.}
    \and 
	Dor Minzer\thanks{Department of Mathematics, Massachusetts Institute of Technology. Supported by NSF CCF award 2227876 and NSF CAREER award 2239160.}}
\date{}

\begin{document}
\maketitle

\begin{abstract}
We study parallel repetition of $k$-player games where the constraints satisfy the {\em projection} property. We prove exponential decay in the value of a parallel repetition of projection games with value less than $1$.
\end{abstract}

\section{Introduction}

We study $k$-player one-round games and the effect on the value of the game when we repeat the game in parallel. 

In a $k$-player game $G$, a verifier chooses $k$ questions $(x^1, x^2, \ldots, x^k)$ from a distribution $\mu$ on the set of questions  $\mathcal{X}_1\times \mathcal{X}_2\times \ldots \times \mathcal{X}_k$ and sends $x^i$ to player $i$. Player $i$ responds to the verifier's question by sending an answer $a^i\in \mathcal{A}_i$ without communicating with the other players. The verifier accepts the answers based on a fixed predicate $V((x^1, x^2,\ldots, x^k), (a^1, a^2,\ldots, a^k))$. The value of the game, denoted by $\val(G)$, is the maximum, over the players' strategies, accepting  probability of the verifier.

The $n$-fold parallel repetition of $G$, denoted by $G^{\otimes n}$, is defined as follows. The verifier sends questions $\vec{x}^i = (x^i_1, x^i_2, \ldots, x^i_n)$ to the $k$ players where for each $j\in [n]$, $(x^1_j, x^2_j,\ldots, x^k_j)$ is sampled from the original distribution $\mu$ independently. The $i^{th}$ player responds with answers $\vec{a}^i\in\mathcal{A}_i^n$. The verifier accepts the answers iff $V((x^1_j, x^2_j,\ldots,  x^k_j), (a^1_j, a^2_j,\ldots, a^k_j)) = 1$ for each $j\in [n]$. 

If $\val(G) = 1$, then it is easy to observe that $\val(G^{\otimes n})$ is also $1$. Also, $\val(G^{\otimes n})\geqslant \val(G)^n$ as the players can achieve value $\val(G)^n$ in the game $G^{\otimes n}$ by simply repeating an optimal strategy for the
game $G$ independently in all the $n$ coordinates. The question of interest is how does the quantity $\val(G^{\otimes n})$ decay with $n$ if the value of the game $G$ is less than $1$?


Verbitsky~\cite{Ver96} showed that for any $k$-player game $G$, if $\val(G)<1$, then $\val(G^{\otimes n}) \leqslant \frac{1}{\alpha(n)}$ where $\alpha(n)$ is an inverse Ackermann function. This result uses the Density Hales-Jewett Theorem~\cite{FurstenbergK89, Polymath12} as a black box.  For $2$-player games, Raz~\cite{Raz98} showed that if $\val(G)<1$, then $\val(G^{\otimes n}) \leqslant 2^{-\Omega_G(n)}$, where we use $\Omega_G(\cdot)$ to clarify that the constant depends on the game $G$. There have been many improvements
that improve the constants in the bounds, and even get better bounds based on the value $\val(G)$
of the initial game~\cite{Hol09, Rao11,DS14, BG14}. These
results on parallel repetition of $2$-player games have found many applications in probabilistically checkable proofs and hardness of approximation~\cite{BellareGS98, Feige98, Hastad01}.

Mittal and Raz~\cite{MittalR21} showed that a strong parallel repetition theorem (i.e., the value of $G^{\otimes n}$ decays exponentially in $n$ in a certain strong sense) for a particular class of
more than $2$-player games implies super-linear lower bounds for Turing machines in the non-uniform model. For any $k\geqslant 2$, Dinur, Harsha, Venkat, and Yuen~\cite{DinurHVY17} showed that for a large class of $k$-player games, called the {\em connected games}, the exponential decay indeed holds. The class of connected games is defined as follows: define the graph $H_G$, whose vertices are the ordered $k$-tuples of questions to the $k$-players, and there is an edge between questions $q$ and $q'$ if they differ in the question to exactly one of the $k$ players, and are the same for the remaining $k-1$ players. The game is said to be connected if the graph $H_G$ is connected. 

A special $3$-player (non-connected) game, called the GHZ Game~\cite{DinurHVY17}, has received much attention. The GHZ game, first introduced by Greenberger, Horne, and Zeilinger~\cite{GreenbergerHZ89}, is a central game in the study of quantum entanglement. Holmgren and Raz~\cite{HolmgrenR20} gave the first polynomial decay in the parallel repetition of the GHZ game. Girish, Holmgren, Mittal, Raz, and Zhan~\cite{GirishHMRZ21_GHZ} later gave a simpler proof of the polynomial decay. Very recently, Braverman, Khot, and Minzer~\cite{BravermanKM22}, using a much simpler proof, improved these previous results and showed an exponential decay in the GHZ game.

Girish, Holmgren, Mittal, Raz, and Zhan~\cite{GirishMRZ22_1} considered the problem of parallel repetition for $3$-player games with binary questions and answers and showed polynomial decay for these games. This was later improved by a subset of the authors \cite{GirishMRZ22_2} to all $3$-player games over binary questions and {\em arbitrary} answer lengths. They also study~\cite{GirishMRZ22_1} {\em player-wise connected} games $G$ that are defined as follows. For each player $i$, define the graph $H_i(G)$, whose vertices are the possible questions for player $i$, and two questions $x$ and $x'$ are connected by an edge if there exists a vector $y$ of questions for all other players, such that both $(x, y)$ and $(x', y)$ are asked by the verifier with non-zero probability. The game $G$ is player-wise connected if, for every $i$, the graph $H_i(G)$ is connected. Girish \etal~\cite{GirishMRZ22_1} showed polynomial decay in the value of $n$-fold parallel repetition of all player-wise connected games. Observe that the notion of player-wise connectedness is more general than the notion of connected games defined above.

In this paper, we will study a special type of $k$-player games, that we refer to as projection games. The formal definition is as follows.

\begin{definition}
    For any $k\geqslant 2$, a $k$-player game $G$ is called a projection game if for every $k$-tuple of question $q = (x^1, x^2,\ldots, x^k)$, there is $D_q\geqslant 1$ and projections $\sigma^i_q: \mathcal{A}_i \rightarrow [D_q]$ for $i\in [k]$, such that $V((x^1, x^2,\ldots, x^k), (a^1, a^2,\ldots, a^k))$ is true iff $\sigma^i_q(a^i) = \sigma^{i'}_q(a^{i'})$ for any $i\neq i'$.
\end{definition}

For every question $q = (x^1, x^2,\ldots, x^k)$, consider a $k$-partite hypergraph $\mathcal{H}_q$ on the vertex set $(\mathcal{A}_1, \mathcal{A}_2, \ldots, \mathcal{A}_k)$ where $(a^1, a^2,\ldots, a^k)$ is an hyperedge if and only if $V((x^1, x^2,\ldots, x^k), (a^1, a^2,\ldots, a^k))$ is true. Then, the projection property means that for every $k$-tuple of questions $q$ in the support of $\mu$, each connected component in $\mathcal{H}_q$ is a complete $k$-partite hypergraph. Note that this definition of projection games is slightly more general than the usual notion of projection $2$-player games~\cite{Raz98, Rao11} where one of the maps $\sigma^i_q$ (either $\sigma^1_q$ or $\sigma^2_q$) is an injective map.

Our main theorem shows that if the value of a projection game $G$ is less than $1$, then the value of $n$-fold parallel repetition of $G$ decays exponentially in $n$.

\begin{theorem}
\label{thm:projection_pr}
    For any $k\geqslant 2$, a projection $k$-player game $G$ and $\eps>0$, if $\val(G) = 1-\eps$, then $\val(G^{\otimes n}) \leqslant \exp({-\Omega_{\eps,G}(n)})$.
\end{theorem}

Projection games are a natural subclass of general games. They played a key role in the development~\cite{AroraLMSS1998, AroraS1998, FGLSS96} of Probabilistically Checkable Proofs (PCPs). In fact, parallel repetition from $2$-player projection games had been useful in proving many~\cite{GHS02,Khot02, Khot02_33, DGKR05, DRS02} tight hardness of approximation results, starting with the work of  Arora, Babai, Stern, and Sweedyk~\cite{AroraBSS97}, Bellare, Goldreich, and Sudan~\cite{BellareGS98}, and H{\aa}stad~\cite{Hastad01}. 

Feige~\cite{Feige98} used a $k$-player projection game, and parallel repetition of the game, to show almost tight hardness of approximating the Set-Cover problem. The decay in the value of a parallel-repeated game, in that case, follows easily from the parallel-repetition theorem for the $2$-player game, as the subgame restricted to any two players has a value less than $1$.

There are $k$-player projection games where the decay in the value of a parallel-repeated game does not trivially follow from the parallel-repetition theorem for the $2$-player game. To give a concrete example, consider a simultaneous Max-3-SAT instance problem defined in~\cite{BhangaleKS15}: the instance consists of $n$ variables $X=\{x_1, x_2, \ldots, x_n\}$ and $k$ instances, $\phi_1, \phi_2, \ldots, \phi_k$, of Max-3-SAT defined over the same set of variables $X$. The verifier chooses a variable $x\in X$ at random and selects clauses $C_i\in \phi_i$ independently such that $x\in C_i$ for all $i\in [k]$. The verifier sends clause $C_i$ to player $i$ and expects a satisfying assignment from $\{0,1\}^3$ to $C_i$ from player $i$. The verifier checks if the assignments returned by the players agree on $x$. Consider the scenario when it is possible to satisfy any $(k-1)$ out of $k$ instances of Max-3-SAT simultaneously, but there is no assignment to $X$ that will satisfy all the $k$ instances. In this case, the value of the game is less than $1$. For any $k'$-player subgame, where $k'<k$, the value of the subgame is $1$. Therefore, we cannot use the parallel repetition of $2$-player games to conclude that the value of $n$-fold parallel repetition of projection games decays with $n$.

\subsection{Proof outline}
As mentioned in the introduction, Dinur, Harsha, Venkat, and Yuen~\cite{DinurHVY17} showed that for any connected $k$-player games $H$ with $\val(H)<1$, the exponential decay holds for the value of $H^{\otimes n}$. 
We start with a $k$-player game $G$ which is not connected to begin with. At a high level, we transform the game $G$ to another game $H$ where $H$ is connected. While doing such a transformation, we want to make sure we have the following two properties.
\begin{enumerate}
    \item If $\val(G)<1$, then $\val(H)<1$.
    \item There is a way to relate $\val(G^{\otimes n})$ with $\val(H^{\otimes n})$, possibly with a small loss in the constants in the exponent.
\end{enumerate}
As $H$ is connected, we have $\val(H^{\otimes n}) = \exp(-\Omega_H(n))$ and this will complete the proof.

There is a trivial transformation that makes any game connected -- add all possible $k$-tuple of questions, play the game $G$ on the original questions, and accept all the newly added questions by default. It is easy to see that if $\val(G)<1$, then the value of the transformed game is less than $1$. However, in this case, there does not seem to be an easy way to relate $\val(G^{\otimes n})$ to the value of $n$-fold parallel repetition of the transformed game.

In order to overcome the issue, we make the game $G$ connected gradually.  More concretely, we start with a game $G_0 = G$ and iteratively, we convert the game $G_\ell$ to $G_{\ell+1}$ for $\ell = 0, 1,\ldots$ with the following three properties.
 \begin{enumerate}
    \item For every $\ell\geqslant 0$, the game $G_\ell$ is a $k$-player game with the questions from the set $\mathcal{X}_1\times \mathcal{X}_2\times \ldots \times \mathcal{X}_k$.
     \item  The game $G_{\ell+1}$ is {\em richer} than the game $G_\ell$. In our case, we would be interested in increasing the support of distribution on questions, i.e., $\mathsf{supp}(\mu(G_{\ell+1}))\supsetneq \mathsf{supp}(\mu(G_\ell))$ (unless, of course, $\mathsf{supp}(\mu(G_\ell)$ is full).
     \item We can relate the value of the game $G_\ell^{\otimes n}$ to the value of the game  $G^{\otimes n}_{\ell+1}$ up to a fixed polynomial factor. Furthermore, $\val(G_{\ell+1})<1$ if $\val(G_{\ell})<1$.
 \end{enumerate}
 Let us see that this is enough to prove our main theorem. Using properties $1$ and $2$, for some $t\geqslant 1$, which only depends on the size of the game $G$, we can conclude that the game $G_t$ has full support and hence is {\em connected}. Using property $3$, we have $\val(G^{\otimes n}) \approx \val(G_t^{\otimes n})^{C_t}$, where $C_t>0$ is a constant that only depends on $t$, and furthermore $\val(G_t)<1$ if $\val(G)<1$ to begin with. Finally, using the result by Dinur, Harsha, Venkat, and Yuen~\cite{DinurHVY17} on connected games, we have, $\val(G_t^{\otimes n}) = \exp(-\Omega_{\eps,G_t}(n))$, and hence $\val(G^{\otimes n}) = \exp(-\Omega_{\eps,t,G}(n))$ if $\val(G)<1$.

\paragraph{The transformation $G_\ell$ to $G_{\ell+1}$.} We illustrate the idea of such a transformation in a $3$-player game $G_\ell$. We start with the game $G_\ell$ and let $\mu(G_\ell)$ be the distributions on the questions in $G_\ell$. For every pair of question-triples $q=(x, y, z)$ and $q'=(x', y', z')$ from $\mathsf{supp}(\mu(G_\ell))$ such that $x=x'$, we add a question triple $\Pi^{3}((q,q')):=(x, y, z')$ to the game $G_{\ell+1}$. Note that in this case, we took two question-triples $(q,q')$ that share player 1's question and generate a question-triple in the new game with the first two players' questions from $q$ and player $3$'s question from $q'$. We now state the set of accepting assignments for $\Pi^{3}((q,q'))$ as follows. If $q\neq q'$, then accept the question $\Pi^{3}((q,q'))$ by default, otherwise accept $\Pi^{3}((q,q'))$ according to the verifier from the original game $G_\ell$ on the question $q(=q')$.\footnote{Note that the way the game $G_{\ell+1}$ is defined, the set of accepting answers for the same question-triple changes based on the underlying pair of questions $(q,q')$, but we ignore this issue in this proof outline.}  We call such a transformation $\mathcal{T}^1_3$ -- the superscript stands for the common player's question from $(q,q')$ and the subscript $3$ stands for taking player 3's question from $q'$ and rest of the questions from $q$ in generating the question-triple in the new game. Succinctly, we write $G_{\ell+1}$ as the game $\mathcal{T}^1_3(G_\ell)$. Likewise, we can define transformations $\mathcal{T}^i_p$ for any $1\leqslant i,p\leqslant 3$.

We show the following key properties of these transformations.
\begin{enumerate}
    \item If $\val(G_\ell)<1$, then $\val(\mathcal{T}^i_p(G_\ell))<1$. Furthermore, $\mathcal{T}^i_p(G_\ell)$ remains a projection game if $G_\ell$ is a projection game.
    \item For every $n\geqslant 1$, $\val(G_\ell^{\otimes n}) \leqslant \val(\mathcal{T}^{i}_{p}(G_\ell)^{\otimes n})^{1/2}$, if $G_\ell$ is a projection game.
\end{enumerate}

The first property is trivial -- in the game $\mathcal{T}^i_p(G_\ell)$, we are still playing the game $G_\ell$ as a subgame and hence its value is less than $1$ if $\val(G_\ell)<1$. For the furthermore part, we are either accepting everything by default or using the same predicate as in the original game and hence this transformation maintains the projection property of the game. 

For the second property, we crucially use the {\bf projection property} of the game $G_\ell$. We show that for any strategy $(\alpha^1, \alpha^2, \alpha^3)$, where $\alpha^i: \mathcal{X}_i^n \rightarrow \mathcal{A}_i^n$, for the game $G_\ell^{\otimes n}$ with value $\eps$, we show that the {\em same} strategy gives value at least $\eps^2$ to the game $\mathcal{T}^i_p(G_\ell)$. We illustrate this for the game $\mathcal{T}^1_3(G_\ell)$. Let $\mu$ be the distribution of questions from $G_\ell$ and $\mu|_{1}$ be the marginal distribution on player 1's questions, we have

\begin{align*}
    \eps^2 &\leqslant \E_{(\x , \y, \z)\sim \mu^{\otimes n}}[V((\x, \y, \z), (\alpha^1(\x), \alpha^2(\y), \alpha^3(\z)))   ] ^2\\
    & = \left(\E_{\v\in \mu^{\otimes n}|_{1}} \left[\E_{\substack{(\x, \y, \z)\sim \mu^{\otimes n},\\ \x = \v}}[V((\x, \y, \z), (\alpha^1(\x), \alpha^2(\y), \alpha^3(\z))) \right]\right)^2\\
    &\leqslant  \E_{\v\in \mu^{\otimes n}|_{1}} \left[\E_{\substack{(\x, \y, \z)\sim \mu^{\otimes n},\\ \x = \v}}[V((\x, \y, \z), (\alpha^1(\x), \alpha^2(\y), \alpha^3(\z))) \right]^2\quad\quad\quad\tag*{(Cauchy-Schwarz)}\\
     &=  \E_{\v\in \mu^{\otimes n}|_{1}} \left[\E_{\substack{(\x, \y, \z)\sim \mu^{\otimes n},\\ (\x', \y', \z')\sim \mu^{\otimes n},\\ \x = \x' = \v}}[V((\x, \y, \z), (\alpha^1(\x), \alpha^2(\y), \alpha^3(\z)))\cdot V((\x', \y', \z'), (\alpha^1(\x'), \alpha^2(\y'), \alpha^3(\z'))) \right]
\end{align*}
Now, if we look at the triple $(\x, \y, \z')$ sampled according to the above distribution, then for each $j\in [n]$, we have that the triple $(x_j, y_j, z'_j)$ is distributed according to the game $\mathcal{T}^1_3(G_\ell)$ independently. In the game, $\mathcal{T}^1_3(G_\ell)$, for any $j\in [n]$ such that $(x_j, y_j, z_j)\neq (x'_j, y'_j, z'_j)$ the new verfier is accepting by default. As for $j\in [n]$ such that $(x_j, y_j, z_j) = (x'_j, y'_j, z'_j)$, the new verifier is accepting according to the original verifier on the question $(x_j, y_j, z_j)$. In this case, suppose $(\alpha^1(\x)_j, \alpha^2(\y)_j, \alpha^3(\z)_j)$ and $(\alpha^1(\x')_j, \alpha^2(\y')_j, \alpha^3(\z')_j)$ are two satisfying assignments to the same question $(x_j, y_j, z_j)$ according to the original game $G_\ell$ with $\alpha^1(\x)_j = \alpha^1(\x')_j$ (as $\x = \x'$), then because of the projection property of the game, we have that $(\alpha^1(\x)_j, \alpha^2(\y)_j, \alpha^3(\z')_j))$ must be a satisfying assignment for $(x_j, y_j, z_j)$. As in the game $\mathcal{T}^1_3(G_\ell)^{\otimes n}$, we are precisely checking this for all such $j\in [n]$, we get that the same strategy $(\alpha^1, \alpha^2, \alpha^3)$ gives
$$\val(\mathcal{T}^1_3(G_\ell)^{\otimes n})\geqslant \eps^2.$$
\paragraph{Putting everything together.} Using the above two properties of the transformations $\mathcal{T}^i_p$, we conclude that if $\val(G)<1$, then for any $t\geq 1$ and vectors $\vec{i}, \vec{p}\in [3]^t$,
$$ \val(G^{\otimes n}) \leqslant \val(\mathcal{T}^{i_{t}}_{p_{t}}(\ldots(\mathcal{T}^{i_2}_{p_2}(\mathcal{T}^{i_1}_{p_1}(G))))^{\otimes n})^{1/2^t}, \mbox{ and } \mathcal{T}^{i_{t}}_{p_{t}}(\ldots(\mathcal{T}^{i_2}_{p_2}(\mathcal{T}^{i_1}_{p_1}(G)))<1.$$
Finally, we show that there exist $t\geq 1$ and vectors $\vec{i}, \vec{p}\in [3]^t$, where $t$ depends on the size of the game $G$, such that the game $\mathcal{T}^{i_{t}}_{p_{t}}(\ldots(\mathcal{T}^{i_2}_{p_2}(\mathcal{T}^{i_1}_{p_1}(G)))$ is connected (in fact, has full support). This implies that 
$$\val(G^{\otimes n}) \leqslant \val(\mathcal{T}^{i_{t}}_{p_{t}}(\ldots(\mathcal{T}^{i_2}_{p_2}(\mathcal{T}^{i_1}_{p_1}(G))))^{\otimes n})^{1/2^t}\leqslant \exp(-\Omega_{t,G}(n)),$$
where the last inequality follows from the result of a parallel repetition theorem~\cite{DinurHVY17} on connected games.

\section{Preliminaries}

We start with a few notations. We use $\mu(G)$ to denote the distribution on the questions in the game $G$. For $i\in [k]$, let $\mu|_{i}$ be the marginal distribution on the questions to player $i$. For a $k$-tuple of questions  $q = (x^1, x^2, \ldots, x^k)$, we denote the question to player $i$ by $q|_{i}$, i.e., $q|_{1} = x^1$, $q|_{2} = x^2$, and so on. For an assignment  $\alpha:=(\alpha^1, \alpha^2, \ldots, \alpha^k)$ to the game $G$, where $\alpha^i: \mathcal{X}_i \rightarrow\mathcal{A}_i$, and any question $q = (x^1, x^2, \ldots, x^k)$, we use the notation $\alpha|_{q}$ to denote the assignment-tuple $(\alpha^1(x^1), \alpha^2(x^2), \ldots, \alpha^k(x^k))$. The size of the game $G$ is referred to as the quantity $k\cdot M\cdot \prod_{i=1}^k |\mathcal{X}_i||\mathcal{A}_i|$. Here, the probability of every atom in $\mathsf{supp}(\mu(G))$ is a multiple of $1/M$, where $M$ is a finite integer (this we assume for simplicity).

\subsection{Parallel repetition of connected games}

As mentioned earlier, Dinur, Harsha, Venkat, and Yuen~\cite{DinurHVY17} showed that for a large class of $k$-player games, called {\em connected games}, the exponential decay indeed holds. Here, we define the notion of connected games for $k$-player games formally.

\begin{definition}[Connected game]
    A game $G$ is called connected if for every two question pairs $(x^1, x^2,\ldots, x^k)$ and $(x'^1, x'^2,\ldots, x'^k)$ from $\mathsf{supp}(\mu(G))$, there is an ordered list of questions  from $\mathsf{supp}(\mu(G))$, $((x^1_\ell, x^2_\ell,\ldots, x^k_\ell))_{\ell=1}^t$ for some $t\geqslant 1$, such that the pairs $((x^1, x^2,\ldots, x^k), (x^1_1, x^2_1,\ldots, x^k_1))$, $((x^1_t, x^2_t,\ldots, x^k_t),(x'^1, x'^2,\ldots, x'^k))$, and $((x^1_\ell, x^2_\ell,\ldots, x^k_\ell), (x^1_{\ell+1}, x^2_{\ell+1},\ldots, x^k_{\ell+1}))$ for all $1\leqslant \ell\leqslant t-1$ differ in only one out of the $k$ questions.
\end{definition}

We will relate the value of $G^{\otimes n}$, where $G$ is a projection game, with a value of $n$-fold parallel repetition of another game $H$ that is connected. The following theorem shows that for connected games with a value less than $1$, the value of repeated games goes down exponentially in $n$.

\begin{theorem}[\cite{DinurHVY17}]
\label{thm:connected_pr}
     For any $k\geqslant 2$ and $\eps>0$, if $H$ is a connected $k$-player game with $\val(H) = 1-\eps$, then $\val(H^{\otimes n}) \leqslant \exp({-\Omega_{\eps,H}(n)})$.
\end{theorem}

\subsection{Variants of multiplayer games}
In this section, we simplify the class of games that we study. Towards this, we define the notion of {\em loosely-connected} games as follows.

\begin{definition}[Loosely-connected game]
    A game on the question set $\mathcal{X}_1\times \mathcal{X}_2\times \ldots \times \mathcal{X}_k$  is loosely-connected if it is not possible to partition $\mathcal{X}_i = \mathcal{X}'_i \cup \mathcal{X}''_i$ for all $i\in [k]$, so that all $k$-tuple of questions from the support of $\mu(G)$ are in $\mathcal{X}'_1\times \mathcal{X}'_2\times \ldots \times \mathcal{X}'_k$ or $\mathcal{X}''_1\times \mathcal{X}''_2\times \ldots \times \mathcal{X}''_k$.
\end{definition}

The following lemma states that we can assume without loss of generality that the game $G$ is {\em loosely-connected.}
\begin{lemma}
    \label{lemma:loosely_c}
    If the exponential decay in the $k$-player parallel repetition holds for all projection loosely-connected games, then it also holds for all projection games.
\end{lemma}
\begin{proof}
   We defer the proof of this lemma to the appendix.
\end{proof}

We also consider a slight variation in the definition of $k$-player games, that we call {\em \rpred}$k$-player games, where we allow a verifier to use a random predicate instead of a fixed predicate during verification. 

\begin{definition}[Random-predicate game]
    A \rpred game $G$ is defined as follows. There exists $R\geqslant 1$ such that the verifier chooses the $k$-tuple of questions $(x^1, x^2,\ldots, x^k)$ according to the distribution $\mu(G)$ on the set of questions and $r\in [R]$ uniformly at random, sends $x^i$ to player $i$. The player $i$ responds with the answer $a^i$. Finally, the verifier accepts the answers based on a fixed predicate $V_r((x^1, x^2,\ldots, x^k), (a^1, a^2,\ldots, a^k))$. We denote such games by $(G, \mu ,[R])$. 
\end{definition}

The following lemma states that for this variation of connected $k$-player games $G$ the exponential decay from~\cite{DinurHVY17} still holds.

\begin{lemma}
\label{lemma:rpred_pr}
     For any connected \rpred $k$-player game $H$ and $\eps>0$, if $\val(H) = 1-\eps$, then $\val(H^{\otimes n}) \leqslant \exp({-\Omega_{\eps,H}(n)})$.
\end{lemma}
\begin{proof}
    We can think of a \rpred $k$-player game $H$ as a $(k+1)$-player game $H'$ as follows. In $H'$, the verifier selects the questions $(x^1, x^2,\ldots, x^k)$ from the game $H$ and $r\in [R]$ uniformly at random. The verifier sends $x^i$ to players $i$ for $i\in [k]$, and sends $r$ to player $k+1$. The player $i\in [k+1]$ responds with the answer $a^i$ ($a^{k+1}$ can be anything). The verifier's predicate in $H'$ is $V((x^1, x^2,\ldots, x^k, r), (a^1, a^2,\ldots, a^k, a^{k+1})) := V_r((x^1, x^2,\ldots, x^k), (a^1, a^2,\ldots, a^k))$.
    
    It is easy to observe that if the game $H'$ is connected then the game $H$ is connected. Furthermore, we have $\val(H^{\otimes n}) = \val(H'^{\otimes n})$ for any $n\geqslant 1$. Using this, the lemma follows from Theorem~\ref{thm:connected_pr}.
\end{proof}

In our proof, we will encounter \rpred games where the choice of the verifier's predicate $V_r$ is not uniform and may depend on the question $q$. The following claims says that we can assume that the distribution on verifier's predicate is uniform form a set of predicates and independent of the questions. This transformation preserves the support of the question-distribution.

\begin{claim}
\label{claim:uniform_rpred}
    Suppose $G$ is a \rpred $k$-player game where on the $k$-tuple of question $q\in \mathcal{X}_1\times \mathcal{X}_2\times \ldots \times \mathcal{X}_k$, the distribution on the verifier's predicate $V_r$ is sampled according to some distribution $\nu_q$ over $[R]$, then there is another game $H$ with the same distribution on the questions as in $G$ such that the verifier for $H$ samples a random predicate $\tilde{V}_m$ where $m\in [M]$ is distributed uniformly over $[M]$, and such that $\val(G^{\otimes n}) = \val(H^{\otimes n})$ for all $n\geqslant 1$. Furthermore, $M$ only depends on the size of the game $G$.
\end{claim}
\begin{proof}
Let $M\in \mathbb{Z}^+$ be a number such that for every question $q$ from the game $G$, each atom from $\mathsf{supp}(\nu_q)$ has probability weight $c/M$ for $1\leqslant c\leqslant M$. In the game $G$, for a question $q$ if $\nu_q(r) = c/M$, then in game $H$, we make $c$ copies $\tilde{V}_{i_1}(q,\cdot), \tilde{V}_{i_2}(q,\cdot), \ldots, \tilde{V}_{i_c}(q,\cdot)$ of the verifier predicate $V_r(q,\cdot)$ for the same question $q$. Thus, for a given question $q$, the verifier in $H$ samples a random $m\in [M]$ and decides based on the predicate $\tilde{V}_m(q,\cdot)$.
\end{proof}

\section{Proof of Theorem~\ref{thm:projection_pr}}

 Throughout this section, we fix a \rpred $k$-player projection game $(G, \mu, [R])$, succinctly written as $G$, on the questions from the set $\mathcal{X}_1\times \mathcal{X}_2\times \ldots \times \mathcal{X}_k$ and let $\mu(G)$ be the distribution on the questions in $G$. Using Claim~\ref{claim:uniform_rpred}, we can assume without loss of generality, that along with the $k$-tuple questions $(x^1, x^2, \ldots, x^k)\sim \mu$, the verifier selects $r\in [R]$ uniformly at random, and after getting answers $(a^1, a^2,\ldots, a^k)$ from the players, applies the predicate $V_r((x^1, x^2, \ldots, x^k)), (a^1, a^2,\ldots, a^k))$. 
 
 The key idea is to use the {\em path-trick} from~\cite{CSP2} to relate the value of the game $G$ (and $G^{\otimes n}$) to another game $H$ (and $H^{\otimes n}$) which is connected.

\subsection{The path-trick and the \texorpdfstring{$i$}{i}-links}

In this section, we define the notion of a {\em link} which is analogous to the notion of the path-trick~\cite{CSP2} that was used in connection to studying dictatorship tests towards showing hardness of approximation of constraint satisfaction problems.

Fix a player $i\in [k]$. An $i$-link from a game $G$ is an ordered pair of original $k$-tuple of questions from $G$ with possible repetition. We will induce the following distribution on $i$-links from $G$.

\begin{itemize}

    \item Pick a question $v$ to player $i$ according to the distribution of $\mu|_{i}$ and sample two $k$-tuple of questions $q = (x^1, x^2, \ldots, x^k)$ and  $q' = (y^1, y^2, \ldots, y^k)$, independently from $\mu$ but conditioned on $x^i = y^i = v$ and output $(q, q')$.
\end{itemize} 
We denote the above distribution on the $i$-links with $L_i(G)$.

To see the utility of $i$-links, the following claim shows that the distribution $L_i(G^{\otimes n})$ on the $i$-links in $G^{\otimes n}$ is the same as the product distribution on the $i$-links from $G$.
\begin{claim}
\label{claim:identical_dist}
    For every game $G$, $i\in [k]$, and $n \geqslant 1$, the following two distributions are identical.
\begin{enumerate}
    \item The distribution $L_i(G^{\otimes n})$.
    \item The distribution $\mathcal{D}_i$ on the $i$-links from $G^{\otimes n}$ defined as follows:
    \begin{itemize}
        \item For each $j\in [n]$, independently sample $(q_j, q'_j)$ from the distribution $L_i(G)$ where $q_j = (x^1_{j}, x^2_{j},\ldots, x^k_{j})$  and $q'_j = (y^1_{j}, y^2_{j},\ldots, y^k_{j})$.
        \item Let $\vec{q} = (q_1, q_2, \ldots, q_n)$ and $\vec{q'} = (q'_1, q'_2, \ldots, q'_n)$. Output $(\vec{q}, \vec{q'})$.
    \end{itemize} 
\end{enumerate}
\end{claim}
\begin{proof}
    Note that $\vec{q}$ and $\vec{q'}$, sampled from $\mathcal{D}_i$,  are the following $k$-tuple of questions
    \begin{align*}
    \begin{array}{cc}
        (x^1_{1}, x^1_{2}, \ldots, x^1_{n}) &   (y^1_{1}, y^1_{2}, \ldots, y^1_{n})\\
        (x^2_{1}, x^2_{2}, \ldots, x^2_{n}) & (y^2_{1}, y^2_{2}, \ldots, y^2_{n})\\
        \vdots & \vdots \\
        \underbrace{(x^k_{1}, x^k_{2}, \ldots, x^k_{n})}_{\vec{q}} &  \underbrace{(y^k_{1}, y^k_{2}, \ldots, y^k_{n})}_{\vec{q'}}\\
        \end{array}
    \end{align*}
    
    For {\em each} $j\in [n]$, the pair of $k$-tuple of questions $(x^1_{j}, x^2_{j}, \ldots, x^k_{j})$ and $(y^1_{j}, y^2_{j}, \ldots, y^k_{j})$ share a common pivot question $p_j = x^i_j = y^i_j$. This means that the question-pair $(\vec{q}, \vec{q'})$ share a common $n$-tuple question $\vec{p}$ from player $i$, where we think of $\vec{q}, \vec{q'}$ as  questions from the game $G^{\otimes n}$. This precisely corresponds to the distribution $L_i(G^{\otimes n})$.
\end{proof}

 Consider the game $G$ the assignments $\alpha^i :\mathcal{X}_i \rightarrow \mathcal{A}_i$ for $i\in [k]$. We say that the link $(q, q')$ from the game $G$ is {\em $r$-consistent} with respect to the global assignments $(\alpha^1, \alpha^2,\ldots, \alpha^k)$ if $q$ as well as $q'$ are satisfied by the predicate $V_r$ on the assignments $(\alpha^1, \alpha^2,\ldots, \alpha^k)$.

\begin{claim}
\label{claim:path}
    Let $n\geqslant 1$, $(\alpha^1, \alpha^2,\ldots, \alpha^k)$  be a strategy for $(G,\mu, [R])$ with $\val(G)\geqslant \eps$. Then with probability at least  $\eps^{2}$, the link $(q,q')$ is $r$-consistent with the assignments $(\alpha^1, \alpha^2,\ldots, \alpha^k)$, where the probability is over $(q,q')$ sampled according to $L_i(G)$ and $r\in [R]$ uniformly at random.
\end{claim}
\begin{proof}
    Fix the provers' strategies $\alpha^i :\mathcal{X}_i \rightarrow \mathcal{A}_i$ for $i\in [k]$ with value at least $\eps$. We have,
$$ \E_{\substack{(x^1, x^2, \ldots, x^k)\sim \mu,\\ r\in [R]}}[ V_r((x^1, x^2,\ldots, x^k), (\alpha^1(x^1), \alpha^2(x^2),\ldots, \alpha^k(x^k))) ] \geqslant \eps.$$
Using the Cauchy-Schwarz inequality, we have,
\begin{align*}
    \eps^2 &\leqslant \E_{\substack{(x^1, x^2, \ldots, x^k)\sim \mu,\\ r\in [R]}}[V_r((x^1, x^2,\ldots, x^k), (\alpha^1(x^1), \alpha^2(x^2),\ldots, \alpha^k(x^k)))   ] ^2\\
    & = \left(\E_{ \substack{v\in \mu|_{i}\\ r\in [R]}} \left[\E_{\substack{(x^1, x^2, \ldots, x^k)\sim \mu\\ x^i = v}}[V_r((x^1, x^2,\ldots, x^k), (\alpha^1(x^1), \alpha^2(x^2),\ldots, \alpha^k(x^k)))\right]\right)^2\\\
\end{align*}
\begin{align*}
    &\leqslant \E_{\substack{v\in \mu|_{i}\\ r\in [R]}} \left[\E_{\substack{(x^1, x^2, \ldots, x^k)\sim \mu\\ x^i = v}}[V_r((x^1, x^2,\ldots, x^k), (\alpha^1(x^1), \alpha^2(x^2),\ldots, \alpha^k(x^k)))]\right]^2\quad\quad\quad\tag*{(Cauchy-Schwarz)}\\
    &=  \E_{\substack{v\in \mu|_{i}\\ r\in [R]}} \left[\E_{\substack{(x^1, x^2, \ldots, x^k)\sim \mu,\\ (y^1, y^2, \ldots, y^k)\sim \mu,\\ x^i = y^i = v}}\left[\begin{array}{c}V_r((x^1, x^2,\ldots, x^k), (\alpha^1(x^1), \alpha^2(x^2),\ldots, \alpha^k(x^k)))\cdot\\ V_r((y^1, y^2,\ldots, y^k), (\alpha^1(y^1), \alpha^2(y^2),\ldots, \alpha^k(y^k)))\end{array}\right] \right].
\end{align*}
 The expression inside the expectation above is precisely the probability that the $i$-link $((x^1, x^2, \ldots, x^k),$ $(y^1, y^2, \ldots, y^k))$ sampled from the distribution $L_i(G)$ is $r$-consistent with respect to the assignments $(\alpha^1, \alpha^2, \ldots, \alpha^k)$. This shows that
$$\Pr_{\substack{(q,q')\sim L_i(G)\\ r\in [R]}} \left[ (q,q') \mbox{ is $r$-consistent with respect to the assignments } (\alpha^1, \alpha^2, \ldots, \alpha^k) \right]\geqslant \eps^2,$$
and this completes the proof.
\end{proof}

\subsection{The transformations \texorpdfstring{$\mathcal{T}^i_p$}{}} 
 We are now ready to define the transformation on the game $G$ that was alluded to at the beginning of this section. We denote the transformed games by $\mathcal{T}^i_p(G)$ for $1\leqslant i,p\leqslant k$. 

The distribution on the $k$-tuple of questions in the game $\mathcal{T}^i_p(G)$ is over $\mathcal{X}_1\times \mathcal{X}_2\times \ldots \times \mathcal{X}_k$ which is defined as follows.
\begin{enumerate}
    \item The verifier samples a link $(q, q')\sim L_i(G)$ where $q = (x^1, x^2, \ldots, x^k)$ and  $q' = (y^1, y^2, \dots, y^k)$
    \item The verifier constructs a $k$-tuple of question by taking $x^p$ from $q$ and $(y^1, \ldots, y^{p-1}, y^{p+1},\dots,  y^k)$ from $q'$. Succinctly, we denote this operation as $(y^1,\ldots,y^{p-1}, x^p, y^{p+1}, \ldots, y^k) := \Pi^p((q,q'))$ (the $p$ stands for taking player $p$'s question from the first question and the remaining players' questions from the second question). 
\end{enumerate}

Before we define the set of satisfying assignments in the transformed game, we first define the set of $r$-consistent assignments, where $r\in [R]$, to an $i$-link. An assignment to a link $(q,q')$ is an assignment to both $q$ and $q'$ (note that each $k$-tuple of question receives a separate assignment from $\mathcal{A}_1\times \mathcal{A}_2\times \ldots\times \mathcal{A}_k$). For an $i$-link $(q,q')$, we define the set of $r$-consistent assignments to $(q,q')$ as follows. An $r$-consistent assignment to an $i$-link $(q,q')$ is an assignment $\sigma$ to the link such that
\begin{itemize}
    \item the verifier accepts $\sigma|_q$ on $q$ and $\sigma|_{q'}$ on $q'$ according to the predicate $V_r$, and
    \item  $\sigma|_{q|_i} = \sigma|_{q'|_i} $, i.e., $\sigma$ gives the same value to the common question $v$ to the player $i$ from $q$ and $q'$.
\end{itemize}
Note that an $i$-link $(q, q')$ may share more than one common question to the players, but the $r$-consistency only cares about the question to player $i$. 

We now define the set of accepting assignments for a $k$-tuple of questions in the \rpred transformed game. The verifier chooses $r\in [R]$ uniformly at random. Suppose a $k$-tuple of questions $\tilde{q} = (z^1, z^2, \ldots, z^k)$ is coming from an $i$-link $(q, q')$, i.e., $\tilde{q} = \Pi^p((q,q'))$, then 
\begin{itemize}
    \item Case 1: If $q\neq q'$, then accept by default.
    \item Case 2: If $q = q'$, then $\tilde{q} = q$. In this case, accept according to the verifier's predicate $V_r$ from the game $G$ on question $\tilde{q}$.
\end{itemize}

Note that the game $\mathcal{T}^i_p(G)$ is a \rpred $k$-player game (the accepting answers for a question $\tilde{q}$ depends on the underlying sampled link as well as the sampled $r\in [R]$), and as described, the distribution on the underlying predicate is not uniform among the set of predicates. However, using Claim~\ref{claim:uniform_rpred}, without loss of generality, we can assume that the underlying distribution on the predicate that the verifier applies in  $\mathcal{T}^i_p(G)$ is uniform from a set of predicates (and independent of the questions). We need this as we will be applying a series of these transformations on the original game, and the transformation above is only defined on games where the verifier's predicate is uniform and independent of the question $q$.

In the transformed game, the verifier either accepts all the answers or accepts answers based on the predicates from the original game $G$. We have the following simple, but important, fact.
\begin{fact}
\label{fact:proj_preserve}
    If $G$ is a \rpred $k$-player projection game, then for every $i,p\in [k]$, the game $\mathcal{T}^i_p(G)$ is also a \rpred projection game.
\end{fact}

\subsubsection{Properties of the transformations \texorpdfstring{$\mathcal{T}^i_p$}{}} 

We start with the first claim that shows that the value of the transformed games is less than $1$ if the value of $G$ is less than $1$.

\begin{claim}
\label{claim:gameGH_ub}
    Fix any $k$-player game $(G, \mu, [R])$. For every $\eps\in (0,1)$ and $i,p \in [k]$, if $\val(G) = 1-\eps$, then $\val(\mathcal{T}^i_p(G))=1-\eps'$ where $\eps'>0$ that depends on $\eps$ and the size of the game $G$.
\end{claim}
\begin{proof}
    First, observe that for every question-tuple $q=(x^1, x^2, \ldots,x^k)$ from the game $G$, the $i$-link $(q,q)$ is present in the support of  $L_i(G)$. Therefore, for every question-tuple $q=(x^1, x^2, \ldots,x^k)$ from the game $\mathcal{T}^i_p(G)$ that is given by the $i$-link $(q, q)$, the verifier of the transformed game selects $r\in [R]$ uniformly at random and uses the predicate $V_r$. Therefore, the transformed game $\mathcal{T}^i_p(G)$ is a convex combination of the original game $G$ and another game $G'$. Let $\tilde{\mu}$ be the distribution on question-tuples in $\mathcal{T}^i_p(G)$, then it can be written as $\tilde{\mu}= \delta\mu + (1-\delta)\mu'$, where $\mu'$ corresponds to the distribution of questions from game $G'$ and $\delta\in (0,1]$ that depends on the size of the game $G$. Thus,
    $$\val(\mathcal{T}^i_p(G)) \leqslant \delta\cdot \val(G) + (1-\delta) = \delta(1-\eps) + (1-\delta)  = 1-\eps\delta < 1.$$
\end{proof}

The following claim relates the value of the original game with the value of the transformed games. This claim crucially uses the fact that the original game $G$ is a projection game.
\begin{claim}
\label{claim:gameGH_lb}
   For any $k$-player projection game $(G, \mu, [R])$, $n\geqslant 1$, and $i,p \in[k]$, $\val(\mathcal{T}^i_p(G)^{\otimes n}) \geqslant \val(G^{\otimes n})^2$.
\end{claim}
\begin{proof}
 As the statement of the claim is symmetric with respect to $p\in [k]$, we prove the claim when $p=1$. For other $p$, the proof is similar.
 
  Using Claim~\ref{claim:identical_dist}, the game $\mathcal{T}^i_1(G)^{\otimes n}$ can be described as follows: Sample an $i$-link $(\vec{q}, \vec{q}')$ from the distribution $L_i(G^{\otimes n})$, sample $\vec{r}\in [R]^n$ uniformly at random, and for every $j\in [n]$ such that $q_j = q'_j$, apply the predicate $V_{r_j}$ for the question $\Pi^1((q_j, q_j))$. 
    
    Let's fix the players' strategy $\alpha:=(\alpha^1,\alpha^2, \ldots, \alpha^k)$ for the game $G^{\otimes n}$ that gives the value $\val(G^{\otimes n})$. For an $i$-link $(\vec{q}, \vec{q}')$ from $G^{\otimes n}$, consider the assignments $\alpha|_{\vec{q}}$ and $\alpha|_{\vec{q}'}$ to the link $(\vec{q}, \vec{q}')$. Using Claim~\ref{claim:path}, for a random $\vec{r}\in [R]^n$ and a randomly selected $i$-link $(\vec{q}, \vec{q}')$, $(\vec{q}, \vec{q}')$ is $\vec{r}$-consistent with respect to the global assignment $\alpha$ with probability at least $\val(G^{\otimes n})^{2}$. When this happens, we show that the assignment $\alpha$ satisfies the constraint on $\Pi^1((\vec{q},\vec{q}'))$ from the game $\mathcal{T}^i_1(G)^{\otimes n}$. Indeed, pick any $j\in [n]$ such that the $j^{th}$ coordinate of the link $(\vec{q}, \vec{q}')$ is $(q_j, q_j)$ (i.e, the same question-tuple). Here, the verifier is using the predicate $V_{r_j}$ on the question $q_j$ in the game $\mathcal{T}^i_1(G)^{\otimes n}$. If the link $(\vec{q}, \vec{q}')$, is $\vec{r}$-consistent with respect to the global assignment $\alpha$, then we have the following
    $$V_{r_j} ((q_j|_1, q_j|_2, \ldots, q_j|_k),(\alpha^1(\vec{q}|_1)_j, \alpha^2(\vec{q}|_2)_j, \ldots, \alpha^k(\vec{q}|_k)_j)) = 1,$$
   $$V_{r_j} ((q_j|_1, q_j|_2, \ldots, q_j|_k),(\alpha^1(\vec{q'}|_1)_j, \alpha^2(\vec{q'}|_2)_j, \ldots, \alpha^k(\vec{q'}|_k)_j)) = 1.$$
   In game $\mathcal{T}^i_1(G)^{\otimes n}$, the verifier is checking the following condition in the coordinate $j$.
   $$V_{r_j} ((q_j|_1, q_j|_2, \ldots, q_j|_k),(\alpha^1(\vec{q}|_1)_j, \alpha^2(\vec{q'}|_2)_j, \ldots, \alpha^k(\vec{q'}|_k)_j)) = 1.$$
   Because $G$ and the predicates $V_r$ satisfy the projection property, and $\vec{q}|_i = \vec{q'}|_i$ because $(\vec{q}, \vec{q'})$ is an $i$-link, we see that if $(\alpha^1(\vec{q}|_1)_j, \alpha^2(\vec{q}|_2)_j, \ldots, \alpha^k(\vec{q}|_k)_j)$ and $(\alpha^1(\vec{q'}|_1)_j, \alpha^2(\vec{q'}|_2)_j, \ldots, \alpha^k(\vec{q'}|_k)_j)$ are the accepting answers for a question according to the predicate $V_{r_j}$, then $(\alpha^1(\vec{q}|_1)_j, \alpha^2(\vec{q'}|_2)_j, \ldots, \alpha^k(\vec{q'}|_k)_j)$ is also an accepting answer for the same question according to the same predicate $V_{r_j}$. This shows that the assignment $\alpha$ passes the verifier's check on all $j\in [n]$ such that  the $j^{th}$ coordinate of the link $(\vec{q}, \vec{q}')$ is $(q_j, q_j)$. For the other coordinates, the game $\mathcal{T}^i_1(G)$ always accepts.

    Hence the same players' strategy $\alpha$ gives $\val(\mathcal{T}^i_1(G)^{\otimes n}) \geqslant \val(G^{\otimes n})^2$.
\end{proof}

Finally, we compose these transformations to get a connected game, starting with a loosely-connected game. Towards this, for any string $\beta\in ([k]\times [k])^m$, where $\beta_j = (\beta^1_j, \beta^2_j)\in [k]\times [k]$, define the transformation $\mathcal{T}^\beta(G)$ as the following transformation
$$\mathcal{T}^{\beta^1_m}_{\beta^2_m}(\ldots(\mathcal{T}^{\beta^1_2}_{\beta^2_2}(\mathcal{T}^{\beta^1_1}_{\beta^2_1}(G)))).$$
For a string $\beta$ of length $m$, define a string $\beta^T$ as a $T$ repeated copy of $\beta$. We have the following claim.
\begin{claim}
\label{claim:larger_conn}
    Let $\beta$ be any permutation of the set $[k]\times [k]$. For large enough $T\geqslant 1$, the game $\mathcal{T}^{\beta^T}(G)$ is connected (in fact, has full support) if $G$ is loosely-connected to begin with. Furthermore, $T\leqslant \prod_{i=1}^k |\mathcal{X}_i|$ which only depends on the size of the game $G$.
\end{claim}
\begin{proof}
    First, any transformation $\mathcal{T}^i_p$ does not shrink the support of the questions of the previous game. By looking closely at the transformation $\mathcal{T}^i_p$, we conclude the following: if $(x^1, x^2, \ldots, x^k)$ and $(y^1, y^2, \ldots, y^k)$ are both in the support of $\mu(G)$ with $x^i = y^i$, then the following question $(y^1,\dots,y^{p-1}, x^p, y^{p+1},\ldots, y^k)$ will be in the support of $\mu(\mathcal{T}^i_p(G))$. Using this, we also observe that if we start with any game $H$ with the property that the series of transformations $\mathcal{T}^\beta(H)$ does not change the support of the questions, then no future transformations will change the support on the questions (as $\beta$ contains every possible transformation $\mathcal{T}^i_p$).

    From the above discussion, we conclude that after some finite (only depends on the size of the game $G$) series of such transformations $\mathcal{T}^{\beta^T}(.)$, the support of the questions does not increase after another series of transformations $\mathcal{T}^{\beta}$. We denote the saturated game by $G_{\mathsf{final}} := \mathcal{T}^{\beta^T}(G)$ and the underlying distribution on the questions by $\mu_{\mathsf{final}}$. We show that $\mu_{\mathsf{final}}$ has full support on $\mathcal{X}_1\times \mathcal{X}_2\times \ldots \times \mathcal{X}_k$. 
    
    Suppose towards a contradiction, the support of $\mu_{\mathsf{final}}$ is not full. For any tuple $q$ of length $k$ and $S\subseteq [k]$, define the tuple $q|_{S}$ by taking the $S$ entries from $q$. Take the {\em smallest} $i\in [k]$ such that there is an $i$-tuple $(z^1, z^2, \ldots, z^i)$ where $(z^1, z^2, \ldots, z^i) \neq q|_{[i]}$ for any $q\in \mathsf{supp}(\mu_{\mathsf{final}})$. Let $z^{\star} := (z^1, z^2, \ldots, z^{i-1})$. For each $i\leqslant \ell \leqslant k$, define the following sets.
    $$\mathcal{S}_{\ell} = \{ x \in \mathcal{X}_{\ell} \mid \exists q\in \mathsf{supp}(\mu_{\mathsf{final}}) \mbox{ s.t. } (z^\star, x) = q|_{[i-1]\cup \{\ell\}}\}. $$
    In other words, $\mathcal{S}_{\ell}$ is a set of all $x \in \mathcal{X}_{\ell}$ that $(z^\star, x)$ can be extended to a valid question tuple from the game $G_{\mathsf{final}}$. Note that by the definition of $z^\star$, $\mathcal{S}_i \subsetneq \mathcal{X}_i$ and furthermore $\mathcal{S}_{i'}\neq \emptyset$ for any $i'\geqslant i$, as by the minimality of $i$, there is a valid question $q$ in $G_{\mathsf{final}}$ such that $q|_{[i-1]}= z^\star$, and hence $q|_{i'}\in \mathcal{S}_{i'}$ for all $i'\geqslant i$.

    \begin{claim}
        There is no $q\in \mathsf{supp}(\mu_{\mathsf{final}})$ such that $q|_{i} \in \overline{\mathcal{S}_i}$ and for some  $i'>i$, $q|_{i'}\in \mathcal{S}_{i'}$.
    \end{claim}
    \begin{proof}
        Suppose there is such a question $q\in \mathsf{supp}(\mu_{\mathsf{final}})$ such that $q|_{i} \in \overline{\mathcal{S}_i}$ and for some  $i'>i$, $q|_{i'}\in \mathcal{S}_{i'}$. Consider a question $q' = (z^1, z^2, \ldots, z^{i-1}, *, *,\ldots)$ where $q'|_{i'} = q|_{i'}$. Note that such a $q'$ is in the $\mathsf{supp}(\mu_{\mathsf{final}})$ from the definition of the set $\mathcal{S}_{i'}$. Furthermore, $(q, q')$ is an $i'$-link in the game $G_{\mathsf{final}}$. Therefore, the question $\Pi^{i}((q, q'))$ will be present in the game $\mathcal{T}^{i'}_{i}(G_{\mathsf{final}})$.  Recall that $\Pi^{i}((q, q')) = (z^1, z^2, \ldots, z^{i-1}, q|_{i}, *,\ldots)$. However, the question $(z^1, z^2, \ldots, z^{i-1}, q|_{i}, *,\ldots)$ is not in $\mathsf{supp}(\mu_{\mathsf{final}})$ as $q|_{i} \in \overline{\mathcal{S}_i}$. This means that the game $G_{\mathsf{final}}$ is not saturated, which is a contradiction.
    \end{proof}
    This claim implies that $\mathcal{S}_{i'} \subsetneq \mathcal{X}_{i'}$ for all $i'>i$ . Indeed, if $\mathcal{S}_{i'}= \mathcal{X}_{i'}$ for some $i'>i$, then the above claim shows that every question $q\in \mathsf{supp}(\mu_{\mathsf{final}})$, $q|_{i} \in \mathcal{S}_i$. Hence, $G_{\mathsf{final}}$ (and hence $G$) is not a loosely-connected game.

    This claim also implies that for every question $q\in \mathsf{supp}(\mu_{\mathsf{final}})$ such that $q|_{i}\in \overline{\mathcal{S}_i}$, we have $q|_{i'}\in \overline{\mathcal{S}_{i'}}$ for every $i'>i$. Consider the partition of the players' question sets $\mathcal{X}_\ell = \mathcal{X}'_\ell \cup \mathcal{X}''_\ell$ such that 
    \begin{itemize}
        \item For all $\ell\leqslant i-1$, $\mathcal{X}'_\ell = \{z^\ell\}$, and $\mathcal{X}''_\ell = \mathcal{X}_\ell\setminus \mathcal{X}'_\ell$, 
        \item for all $t\geqslant \ell$, $\mathcal{X}'_\ell = \mathcal{S}_\ell$, and $\mathcal{X}''_\ell = \mathcal{X}_\ell\setminus \mathcal{X}'_\ell$.
    \end{itemize}
    
    Because $G$ (and hence $G_{\mathsf{final}}$) is loosely connected, there must be a  question $q\in \mathsf{supp}(\mu_{\mathsf{final}})$ such that $q|_{i'}\in \overline{\mathcal{S}_{i'}}$ for every $i'\geqslant i$ and $q|_t = z^t$ for some $1\leqslant t\leqslant i-1$. Consider a question $q' = (z^1, z^2, \ldots, z^{i-1}, *, *,\ldots)$ such that $q'$ is in the $\mathsf{supp}(\mu_{\mathsf{final}})$. Now, the pair of questions $(q', q)$ is an $t$-link in the game $G_{\mathsf{final}}$, furthermore, for a questions $q'' : = \Pi^{i}((q,q'))$, we have $q'' = (z^1, z^2, z^3, \ldots, z^{i-1}, q|_{i}, \ldots)$ where $q|_{i}\in  \overline{\mathcal{S}_{i}}$. As $G_{\mathsf{final}}$ is saturated, $q'''\in \mathsf{supp}(\mu_{\mathsf{final}})$ but this contradicts the definition of $\mathcal{S}_i$.
\end{proof}

\subsection{Finishing the proof}
Let us see why these claims above are enough to prove Theorem~\ref{thm:projection_pr}.
\paragraph{Proof of Theorem~\ref{thm:projection_pr}.}  We start with a projection game $G$ with $\val(G) = 1-\eps$ for some $\eps>0$. First, using Lemma~\ref{lemma:loosely_c}, we can assume without loss of generality that $G$ is loosely-connected. Let $(\vec{i},\vec{p})\in ([k]\times [k])^{Tk^2}$ with $(\vec{i},\vec{p}) = \beta^T$,  i.e., $(i_t, p_t)$ is the $t^{th}$ entry from the string $\beta^T$, where $\beta$ and $T$ are from Claim~\ref{claim:larger_conn}. Let $T' = T\cdot k^2$.  \\
Using Claim~\ref{claim:gameGH_lb} on the projection game $G$ with $i_1, p_1\in [k]$, we have
$$ \val(G^{\otimes n}) \leqslant \val(\mathcal{T}^{i_1}_{p_1}(G)^{\otimes n})^{1/2}.$$
Fact~\ref{fact:proj_preserve} shows that $\mathcal{T}^{i_1}_{p_1}(G)$ is a projection game, and hence applying Claim~\ref{claim:gameGH_lb} on the projection game $\mathcal{T}^{i_1}_{p_1}(G)$ with $i_2, p_2\in [k]$, we get
$$ \val(\mathcal{T}^{i_1}_{p_1}(G)^{\otimes n}) \leqslant \val(\mathcal{T}^{i_2}_{p_2}(\mathcal{T}^{i_1}_{p_1}(G))^{\otimes n})^{1/2}.$$
Repeating this process $T'$ times, we get
$$ \val(G^{\otimes n}) \leqslant \val(\mathcal{T}^{i_{T'}}_{p_{T'}}(\ldots(\mathcal{T}^{i_2}_{p_2}(\mathcal{T}^{i_1}_{p_1}(G))))^{\otimes n})^{1/2^{T'}}.$$
Using Claim~\ref{claim:gameGH_ub} repeatedly, we have 
$$\val(\mathcal{T}^{i_{T'}}_{p_{T'}}(\ldots(\mathcal{T}^{i_2}_{p_1}(\mathcal{T}^{i_1}_{p_1}(G)))))\leqslant 1-\eps',$$
where $\eps'>0$, that only depends on $\eps$, $T'$, and the size of the game $G$. Finally, using Claim~\ref{claim:larger_conn}, we have that the game $\mathcal{T}^{i_{T'}}_{p_{T'}}(\ldots(\mathcal{T}^{i_2}_{p_2}(\mathcal{T}^{i_1}_{p_1}(G))))$ is connected and hence by Lemma~\ref{lemma:rpred_pr},
$$\val(\mathcal{T}^{i_{T'}}_{p_{T'}}(\ldots(\mathcal{T}^{i_2}_{p_2}(\mathcal{T}^{i_1}_{p_1}(G))))^{\otimes n})\leqslant \exp(-\Omega_{\eps,T,G}(n)).$$
Overall, we get $\val(G^{\otimes n}) \leqslant  \exp(-\Omega_{\eps,T,G}(n))$ and the proof is completed as $T$ only depends on the size of the original game $G$.

\paragraph{Acknowledgement.} We thank Kunal Mittal for helpful discussions at the early stage of this work.

\bibliographystyle{alpha}
\bibliography{refs}

\appendix

\section{Proof of Lemma~\ref{lemma:loosely_c}}
 Let $G$ be any projection game that is not loosely-connected with $\val(G) = 1-\eps$ for some $\eps>0$. Without loss of generality, we can assume that there are partitions $\mathcal{X}_i = \mathcal{X}'_i \cup \mathcal{X}''_i$ for all $i\in [k]$ such that all the questions from the support of $\mu(G)$ are from $\mathcal{X}'_1\times \mathcal{X}'_2\times \ldots \times \mathcal{X}'_k$ or $\mathcal{X}''_1\times \mathcal{X}''_2\times \ldots \times \mathcal{X}''_k$, and furthermore, the game restricted to $\mathcal{X}'_1\times \mathcal{X}'_2\times \ldots \times \mathcal{X}'_k$ (call it $G'$) and $\mathcal{X}''_1\times \mathcal{X}''_2\times \ldots \times \mathcal{X}''_k$ (call it $G'')$ are loosely-connected individually. The verifier's distribution $\mu(G)$ on the question-tuples can be thought of as $\mu = (1-\delta) \mu' + \delta \mu''$ where the support of $\mu'$ is from  $\mathcal{X}'_1\times \mathcal{X}'_2\times \ldots \times \mathcal{X}'_k$ and the support of $\mu''$ is from $\mathcal{X}''_1\times \mathcal{X}''_2\times \ldots \times \mathcal{X}''_k$.
    
    Now, since the value of the game $G$ is at most $1-\eps$, we have $\min\{\val(G'), \val(G'')\}\leqslant 1-\eps$. Without loss of generality, suppose we have $\val(G')\leqslant 1-\eps$. Let the value of $G^{\otimes n}$ be $\eta$. We will show that the value of the game $G'^{\otimes n'}$ is also at least $\eta- 2^{-\Omega_\delta(n)}$ for some $n'=\Omega_\delta(n)$. This will finish the proof of the lemma as we have $\val(G'^{\otimes n'})\leqslant \exp(-\Omega_{\eps}(n'))$ using the fact that $G'$ is a loosely-connected projection game.

    Fix a strategy $(\alpha^1, \alpha^2, \ldots, \alpha^k)$ for $G^{\otimes n}$ with value $\eta$. The $k$-tuple questions from $G$ can be alternatively sampled as follows. First sample a set $T\subseteq [n]$ by adding $i\in T$ independently with probability $(1-\delta)$. Then for each $i\in T$, sample a $k$-tuple question from the distribution $\mu'$ independently. Similarly, for each $i\notin T$, sample a $k$-tuple question from the distribution $\mu''$ independently. We have,

    $$\E_{\substack{T\subseteq_{1-\delta} [n]\\ (\vec{x}^1|_{\overline{T}}, \vec{x}^2|_{\overline{T}},\ldots, \vec{x}^k|_{\overline{T}}) \sim \mu''^{\otimes |\overline{T}|}}} \E_{(\vec{x}^1|_{T}, \vec{x}^2|_{T}, \ldots, \vec{x}^k|_{T}) \sim \mu'^{\otimes |T|}}\left[ V((\vec{x}^1, \vec{x}^2, \ldots, \vec{x}^k) , (\alpha^1(\vec{x}^1), \alpha^2(\vec{x}^2),\ldots,  \alpha^k(\vec{x}^k)) )\right]  = \eta.$$ 
    
    From this, by the Chernoff Bound and an averaging argument, it follows that there exists $T\subseteq [n]$ such that $|T| \geqslant (1-\delta)n/2$ and $(\vec{y}^1, \vec{y}^2,\ldots,  \vec{y}^k) \sim \mu''^{\otimes |\overline{T}|}$ such that 
    \begin{align*}
  \E_{(\vec{x}^1, \vec{x}^2, \ldots, \vec{x}^k) \sim \mu'^{\otimes |T|}}\left[ V(((\vec{y}^1 ,\vec{x}^1) , (\vec{y}^2 ,\vec{x}^2) , \ldots, (\vec{y}^k ,\vec{x}^k) ) , (\alpha^1((\vec{y}^1 ,\vec{x}^1)), \alpha^2((\vec{y}^2 ,\vec{x}^2)), \alpha^k((\vec{y}^k ,\vec{x}^k))) )\right]\\
  \geqslant \eta - 2^{-\Omega_\delta(n)}.      
    \end{align*} 
    Here, the string $(\vec{y} ,\vec{x})$ is formed by plugging $\vec{y}$ in the coordinates $\overline{T}$ and $\vec{x}$ in the coordinates $T$. The distribution on the questions in the expectation above precisely corresponds to the game $G'^{\otimes |T|}$. Thus, the strategy $\alpha'^i(\vec{x}) := \alpha^i(\vec{y} ,\vec{x})$ for all $i\in [k]$ gives the value at least $\eta - 2^{-\Omega_\delta(n)}$ for the game $G'^{\otimes |T|}$.
\end{document}